\documentclass[runningheads,a4paper]{llncs}

\usepackage{amssymb}
\usepackage{graphicx}

\usepackage{amsmath}
\usepackage{algorithm}
\usepackage{algorithmic}

\usepackage{url}
\newcommand{\keywords}[1]{\par\addvspace\baselineskip
\noindent\keywordname\enspace\ignorespaces#1}

\begin{document}

\mainmatter  

\title{Approximating Scheduling Machines with Capacity Constraints}

\subtitle{\scriptsize{(This is a correction of paper at FAW2009)}}


\author{Chi Zhang
\thanks{\email{arxor.kid[AT]gmail.com}. Supported in part by the National High Technology Research and Development Program of China (2008AA01Z401), NSFC of China (90612001), SRFDP of China (20070055054), and Science and Technology Development Plan of Tianjin (08JCYBJC13000).}
\and Gang Wang \and Xiaoguang Liu \and Jing Liu}

\institute{Nankai-Baidu Joint Lab, College of Information Technical Science, Nankai University}

%
%

\maketitle

\begin{abstract}
In the Scheduling Machines with Capacity Constraints problem, we are given $k$ identical machines, each of which can process at most $m_i$ jobs. $M$ jobs are also given, where job $j$ has a non-negative processing time length $t_j \geq 0$. The task is to find a schedule such that the makespan is minimized and the capacity constraints are met. In this paper, we present a $3$-approximation algorithm using an extension of \emph{Iterative Rounding Method} introduced by Jain \cite{Jai:01}. To the best of the authors' knowledge, this is the first attempt to apply \emph{Iterative Rounding Method} to scheduling problem with capacity constraints.
\end{abstract}

\section{Introduction}

We consider the Scheduling Machines with Capacity Constraints problem (SMCC): There are $k$ identical machines, and machine $i$ can process at most $m_i$ jobs. Given $M \leq \sum_{1 \leq i \leq k}{m_i}$ jobs with their processing time lengths, we are to find a schedule of jobs to machines that minimizes the makespan and meets the capacity constraints.

Scheduling problem is a classical $\mathbf{NP}$-Hard problem and has been studied extensively. In the general setting, we are given set $T$ of tasks, number $k$ of machines, length $l(t,i) \in \mathbf{Z}^+$ for each $t \in T$ and machine $i \in [1..k]$, the task is to find a schedule for $T$, namely, a function $f: T \rightarrow [1..k]$, to minimize $\max_{i \in [1..k]}{\sum_{t \in T, f(t) = i}{l(t,i)}}$\enspace. Lenstra, Shmoys and Tardos \cite{Len:Shm:Tar:90} gave a $2$-approximation algorithm for the general version and proved that for any $\epsilon > 0$ no $(\frac{3}{2}-\epsilon)$-approximation algorithm exists unless $\mathbf{P}=\mathbf{NP}$\enspace. Their method based on applying rounding techniques on fractional solution to linear programming relaxation. Gairing, Monien and Woclaw \cite{Gai:Mon:Woc:07} gave a faster combinatorial $2$-approximation algorithm for the general problem. They replaced the classical technique of solving the LP-relaxation and rounding afterwards by a completely integral approach. For the variation in which the number of processors $k$ is constant, Angel, Bampis and Kononov \cite{Ang:Bam:Kon:01} gave a \emph{fully polynomial-time approximation scheme (FPTAS)}. For the uniform variation where $l(t,i)$ is independent of the processor $i$, Hochbaum and Shmoys \cite{Hoc:Shm:87} gave a \emph{polynomial-time approximation scheme (PTAS)}.

The SMCC problem is one of the uniform variations, with capacity constraints on machines. One special case of SMCC problem in which there are only two identical machines was studied in \cite{Tsai:92} \cite{Yang:Ye:Zhang:03} \cite{Zhang:Ye:01}. Woeginger \cite{Woe:05} gave a FPTAS for the same problem. General SMCC problem is a natural generalization of scheduling problem without capacity constraints and can be used in some applications in real world, such as students distributions in university, the Crew Scheduling problem in Airlines Scheduling \cite{Ets:Mat:85} \cite{Rus:Hof:Pad:95}, etc. In the Crew Scheduling problem, crew rotations, sequences of flights legs to be flown by a single crew over a period of a few days, are given. Crews are paid by the amount of flying hours, which is determined by the scheduled rotations. Airline company wants to equalize the salaries of crews, i.e. to make the highest salary paid to crews minimum. Rotations starts and ends at the same crew base and must satisfy a large set of work rules based labor contracts covering crew personnel. In the concern of safety issues, one common contract requirement is the maximum times of flying of a single crew in a period of time. So the aim is to find a scheduling of rotations to crews that minimizes the highest salary and meets the maximum flying times constraints.

In many literature, researchers approached scheduling problem using rounding techniques. Lenstra, Shmoys and Tardos \cite{Len:Shm:Tar:90} applied rounding method to the decision problem to derive a \emph{$\rho$-relaxed decision procedure} and then used a binary search to obtain an approximation solution. In the SMCC problem, the capacity constraints defeat many previous methods. In this paper, our algorithm is one of the class of rounding algorithms, but use a different rounding method introduced by Jain \cite{Jai:01}. We do not round off the whole fractional solution in a single stage. Instead, we round it off iteratively.

\emph{Iterative Rounding Method}, introduced by Jain \cite{Jai:01}, was used in his breakthrough work on the Survivable Network Design problem. This rounding method does not need the half-integrality, but only requires that at each iteration there exist some variables with bounded values. In \cite{Jai:01}, Jain observed that at each iteration one can always find a edge $e$ has $x_e$ at least $1/2$, which ensures that the algorithm has an approximation ratio of $2$\enspace. As a successful extension of Jain's method, Mohit Singh and Lap Chi Lau \cite{Moh:Lap:07} considered the Minimum Bounded Degree Spanning Trees problem and gave an algorithm that produces a solution, which has at most the cost of optimal solution while violating vertices degrees constraints by $1$ at most. As far as the authors know, \emph{Iterative Rounding Method} has been used in graph problems, and has produced many beautiful results.

In this paper, we apply \emph{Iterative Rounding Method} to the scheduling problem with capacity constraints and obtain a $3$-approximation algorithm. To the best of the authors' knowledge, this is the first attempt to approach scheduling problem with capacity constraints using \emph{Iterative Rounding Method}.

The rest of the paper is organized as follows. In Section \ref{preliminary}, we formulate the SMCC problem as an Integer Program, give its natural relaxation and introduce our relaxation, Bounded Linear Programming Relaxation (BLPR). In Section \ref{techniques}, we present some properties of BLPR and prove theorems that support our algorithm. In Section \ref{approximation_algorithm}, we present bounding theorems and an approximation algorithm, $IRA$, and prove that it has an approximation ratio of $3$\enspace.

\section{Preliminary}\label{preliminary}

Formally, the SMCC problem is as follows: Given a positive integer $k$, $k$ positive integers $\{m_i|m_i > 0, 1 \leq i \leq k\}$, $M$ non-negative integers $\{t_j|t_j \geq 0, 1 \leq j \leq M \leq \sum_{i = 1}^{k}{m_i}\}$, we are to solve the following Integer Program (IP):
\begin{equation*}
  \begin{array}{rrlr}
    \mbox{minimize } & c &&\\
    \mbox{subject to } & \sum_{j = 1}^{M}{x_{ij}t_j} - c &\leq 0 & 1 \leq i \leq k\\
    & \sum_{j = 1}^{M}{x_{ij}} & \leq m_i & 1 \leq i \leq k\\
    & \sum_{i = 1}^{k}{x_{ij}} & = 1 & 1 \leq j \leq M\\
    & x_{ij} & \in \{0,1\} & 1 \leq i \leq k, 1 \leq j \leq M\\
  \end{array}
\end{equation*}

There are some relaxations, one of which is the following natural Linear Programming Relaxation (LPR) dropping the integrality constraints.
\begin{equation*}
  \begin{array}{rrlr}
    \mbox{minimize } & c &&\\
    \mbox{subject to } & \sum_{j = 1}^{M}{x_{ij}t_j} - c & \leq 0 & 1 \leq i \leq k\\
    & \sum_{j = 1}^{M}{x_{ij}} & \leq m_i & 1 \leq i \leq k\\
    & \sum_{i = 1}^{k}{x_{ij}} & = 1 & 1 \leq j \leq M\\
    & x_{ij} & \geq 0 & 1 \leq i \leq k, 1 \leq j \leq M\\
  \end{array}
\end{equation*}

We don't use LPR directly, but use an alternative relaxation, Bounded Linear Programming Relaxation (BLPR): Given a positive integer $k$, $k$ positive integers $\{m_i|m_i > 0, 1 \leq i \leq k\}$, $M$ non-negative integers $\{t_j|t_j \geq 0, 1 \leq j \leq M \leq \sum_{i = 1}^{k}{m_i}\}$, a real vector $b = (b_1, b_2, \dots, b_k)$ and $\mathcal{F} \subseteq \{(i,j)| 1 \leq i \leq k, 1 \leq j \leq M\}$, find a feasible solution under the following constraints
\begin{align}
  &\mbox{subject to } & \sum_{j = 1}^{M}{x_{ij}t_j} & \leq b_i & 1 \leq i \leq k& \label{first_constraints}\\
  && \sum_{j = 1}^{M}{x_{ij}} & \leq m_i & 1 \leq i \leq k& \label{second_constraints}\\
  && \sum_{i = 1}^{k}{x_{ij}} & = 1 & 1 \leq j \leq M& \label{third_constraints}\\
  && x_{ij} & = 1 & (i,j) \in \mathcal{F}&\\
  && x_{ij} & \geq 0 & 1 \leq i \leq k, 1 \leq j \leq M& \label{non_neg_constraints}
\end{align}
where vector $b = (b_1, b_2, \dots, b_k)$, called \emph{upper bounding vector}, is added to depict the different upper bounds of machines more precisely, and $\mathcal{F}$ is added to represent the partial solution in algorithm. Each $(i,j) \in \mathcal{F}$ indicates that job $j$ has been scheduled to machine $i$\enspace. Those $\{x_{ij}|(i,j) \in \mathcal{F}\}$ are considered as constants.

We will show that properly constructing vector $b = (b_1, b_2, \dots, b_k)$ makes the solution produced by our algorithm under control and easy to analyze.

\begin{definition}
  In a BLPR problem $\Lambda$, \emph{upper bounding vector} $b = (b_1, b_2, \dots, b_k)$ is called \emph{feasible} if $\Lambda$ is feasible.
\end{definition}

Keeping the \emph{upper bounding vector} $b$ \emph{feasible} all the time is the key of our algorithm, which guarantees that we can always find a feasible solution bounded by $b$\enspace.

\section{Techniques}\label{techniques}

Before we present our algorithm, we need to introduce some properties of BLPR.

With respect to the partial solution $\mathcal{F}$, let $c_i$ denote the number of already scheduled jobs in machine $i$, namely, $c_i = |\{(i,j)|(i,j) \in \mathcal{F}\}|$\enspace. Note that $m_i - c_i$ indicates the \emph{free} capacity in machine $i$\enspace.

We call a job \emph{free} if it has not been scheduled to any machine and call a machine \emph{free} if it still has free capacity. For a feasible fractional solution $x$ to $\Lambda$, define a bipartite graph $G(x) = G(L,R,E)$, called \emph{supporting graph}, where $L$ represents the set of \emph{free} machines, $R$ represents the set of \emph{free} jobs and $E = \{(i,j)|x_{ij} > 0, (i,j) \notin \mathcal{F}\}$\enspace. We denote the number of \emph{free} jobs and the number of \emph{free} machines by $M^*$ and $k^*$ respectively. Note that for \emph{free} job $j$, $\sum_{(i,j) \in E}{x_{ij}} = 1$\enspace.

Consider the \emph{Constraint Matrix} of $\Lambda$, which consists of the coefficients of the left side of equalities and inequalities, except for the non-negativity constraints from \eqref{non_neg_constraints}:
\begin{equation}
  \begin{pmatrix}
    \begin{pmatrix}t_1&t_2&\dots&t_M\end{pmatrix} &&&\\
    &\begin{pmatrix}t_1&t_2&\dots&t_M\end{pmatrix} &&\\
    &&\ddots&\\
    &&&\begin{pmatrix}t_1&t_2&\dots&t_M\end{pmatrix}\\
    \begin{pmatrix}1&1&\dots&1\end{pmatrix} &&&\\
    &\begin{pmatrix}1&1&\dots&1\end{pmatrix} &&\\
    &&\ddots&\\
    &&&\begin{pmatrix}1&1&\dots&1\end{pmatrix}\\
    \begin{pmatrix}
      1&&&\\
      &1&&\\
      &&\ddots&\\
      &&&1
    \end{pmatrix} &
    \begin{pmatrix}
      1&&&\\
      &1&&\\
      &&\ddots&\\
      &&&1
    \end{pmatrix} &
    \dots &
    \begin{pmatrix}
      1&&&\\
      &1&&\\
      &&\ddots&\\
      &&&1
    \end{pmatrix}
  \end{pmatrix}
\end{equation}
where the $1^{\mbox{st}}$ to $k^{\mbox{th}}$ rows represent the constraints from \eqref{first_constraints}, the $(k+1)^{\mbox{th}}$ to $(2k)^{\mbox{th}}$ rows represent the constraints from \eqref{second_constraints}, and the $(2k+1)^{\mbox{th}}$ to $(2k+M)^{\mbox{th}}$ rows represent the constraints from \eqref{third_constraints}.

One can verify that the $(k+1)^{\mbox{th}}$ row can be linearly expressed by the rest of rows. Thus the rank of \emph{Constraints Matrix} is bounded by the following lemma
\begin{lemma}\label{rank_bounding}
  \emph{Constraints Matrix} has a rank at most $M + 2k - 1$\enspace.
  \qed
\end{lemma}

Recall that, a basic solution $x$ to $\Lambda$ is the unique solution determined by a set of linearly independent \emph{tight constraints} that are satisfied as equalities. We remove all zero variables in $x$ so that no \emph{tight constraints} comes from \eqref{non_neg_constraints}. Thus the number of non-zero variables in $x$ never exceeds the rank of \emph{Constraints Matrix}. When $\mathcal{F} = \emptyset$, the following inequality holds
\begin{equation}
  |E| \leq M + 2k - 1
\end{equation}

We can remove those non-\emph{free} machines from $\Lambda$, move fixed variables $\{x_{ij}|(i,j) \in \mathcal{F}\}$ to the right side of the equalities and inequalities as constants and remove variables fixed to $0$. By doing this, we obtain a new sub-problem and only focus on \emph{free} jobs and \emph{free} machines. In the new sub-problem, Lemma \ref{rank_bounding} holds. So in general, the following corollary holds.
\begin{corollary}\label{edges_bounding}
  Given a BLPR problem, $\Lambda$, its basic solution $x$ and \emph{supporting graph} $G(x) = G(L,R,E)$, we have
  \begin{equation*}
    |E| \leq M^* + 2k^* - 1
  \end{equation*}
  \qed
\end{corollary}

We introduce lemmas on the basic solution to $\Lambda$ when there are no less \emph{free} jobs than twice the \emph{free} machines, namely, $M^* \geq 2k^*$.

\begin{lemma}
  If $\Lambda$ is feasible with $M^* \geq 2k^*$ and $x$ is a basic solution, there exist $M^* - 2k^* + 1$ variables with values of $1$\enspace.
\end{lemma}
\begin{proof}
  For a basic solution $x$, we construct \emph{supporting graph} $G(x) = G(L,R,E)$\enspace. Suppose that $l$ of $M^*$ \emph{free} jobs are of degree of at most one in $G$. Note that each of them has degree at least one. Each of the rest $M^* - l$ \emph{free} jobs has degree of more than one. The following inequality holds.
  \begin{equation}
    |E| \geq 2(M^* - l) + l = 2M^* - l
  \end{equation}

  By Corollary \ref{edges_bounding}, we have
  \begin{equation}
    l \geq M^* - 2k^* + 1
  \end{equation}
  \qed
\end{proof}

The following corollary holds
\begin{corollary}\label{iterative_theorem}
  If $\Lambda$ is feasible with $M^* \geq 2k^*$ and $x$ is a basic solution, there exist a \emph{free} machine $p$ and a \emph{free} job $q$ such that $x_{pq} = 1$\enspace.
  \qed
\end{corollary}

\section{A $3$-approximation Algorithm}\label{approximation_algorithm}

In this section, we present an approximation algorithm $IRA$\enspace. Let $\mathcal{A}$ denote the makespan in the solution produced by $IRA$, $\mathcal{OPT}$ denote the makespan in the optimal solution to $\Delta$\enspace.

We introduce three bounding theorems on BLPR. Noting that when $M^* \geq 2k^*$ we can find a $x_{pq} = 1$, we can schedule $q$ to $p$ without increasing the lengths in the fractional solution. We first show the theorem for the case $M^* \geq 2k^*$.
\begin{theorem}\label{main_theorem1}
  Given a BLPR problem, $\Lambda$, with $M^* \geq 2k^*$ and its basic solution $x$\enspace. Based on $\Lambda$, we construct a new BLPR problem $f(\Lambda)$ as follows:
  \begin{enumerate}
  \item Find a variable $x_{pq} = 1$\enspace;
  \item $\mathcal{F}' \gets \mathcal{F} + (p,q)$\enspace;
  \item The rest parts of $f(\Lambda)$ are the same as $\Lambda$\enspace.
  \end{enumerate}

  If $b$ is a \emph{feasible upper bounding vector} of $\Lambda$ then $b'$ is a \emph{feasible upper bounding vector} of $f(\Lambda)$\enspace.
  \qed
\end{theorem}

When $M^* < 2k^*$ and some \emph{free} machines have \emph{free} capacity of one, we have the following theorem.
\begin{theorem}\label{main_theorem2}
  Given a BLPR problem, $\Lambda$, with $M^* < 2k^*$ and its basic solution $x$. Moreover some \emph{free} machines have \emph{free} capacity of $1$. Based on $\Lambda$, we construct a new BLPR problem $g(\Lambda)$ as follows:
  \begin{enumerate}
  \item Let $p$ denote a machine with $1$ \emph{free} capacity;
  \item Let $q$ denote the \emph{free} job with the largest length;
  \item $\mathcal{F}' \gets \mathcal{F} + (p,q)$\enspace;
  \item $b' \gets (b_1, b_2, \dots, b_{p-1}, b_p + t_q, b_{p+1}, \dots, b_k)$\enspace;
  \item The rest parts of $g(\Lambda)$ are the same as $\Lambda$\enspace.
  \end{enumerate}

  If $b$ is a \emph{feasible upper bounding vector} of $\Lambda$ then $b'$ is a \emph{feasible upper bounding vector} of $g(\Lambda)$\enspace.
\end{theorem}
\begin{proof}
  To schedule $q$ to $p$, for each \emph{free} machine $s \neq p$ with $x_{sq} > 0$, we move $x_{sq}$ fraction of job $q$ to machine $p$ then move back from $p$ as much as possible but no more than $x_{sq}$ fraction of jobs other than $q$ as in Algorithm \ref{g}.

  \begin{algorithm}
    \caption{$G-transition$}\label{g}
    \begin{algorithmic}[1]
      \REQUIRE A BLPR $\Lambda$ with $M^* < 2k^*$ and machine $p$ has \emph{free} capacity of $1$
      \ENSURE A BLPR $g(\Lambda)$
      \STATE Let $x$ be a basic feasible solution to $\Lambda$\enspace and $q$ be the longest \emph{free} job;
      \WHILE {there exists $p' \neq p$ such that $x_{p'{q}} > 0$}
        \IF {there exists $q' \neq q$ such that $x_{p{q'}} > 0$}
          \STATE $\alpha \gets \min\{x_{p'{q}}, x_{p{q'}}\}$\enspace;
          \STATE $x_{p'{q}} \gets x_{p'{q}} - \alpha$\enspace;
          \STATE $x_{p{q'}} \gets x_{p{q'}} - \alpha$\enspace;
        \ELSE [$\forall q' \neq q, x_{p{q'}} = 0$]
          \STATE $x_{p{q'}} \gets x_{p{q'}} + x_{p'{q}}$\enspace;
        \ENDIF
      \ENDWHILE
    \end{algorithmic}
  \end{algorithm}

  Because $q$ has the largest length among the \emph{free} jobs, we can guarantee that the length of each \emph{free} machine $p' \neq p$ will not increase and the length of machine $p$ will increase by at most $t_q$\enspace. Note that in $g(\Lambda)$\enspace, $x_{pq} = 1$ and $p$ is no longer \emph{free}. This implies there is a feasible solution to $g(\Lambda)$\enspace.
  \qed
\end{proof}

When $M^* < 2k^*$ and every \emph{free} machine has more than $1$ \emph{free} capacity, we can schedule these $M^*$ jobs arbitrarily but only assuring that each machine gets no more than $2$ jobs. One can prove the following theorem.

\begin{theorem}\label{main_theorem3}
  Given a feasible BLPR problem, $\Lambda$, with $M^* < 2k^*$. Moreover every \emph{free} machine has \emph{free} capacity of at least $2$. We construct a new BLPR problem $h(\Lambda)$ as follows:
  \begin{enumerate}
  \item Schedule \emph{free} jobs arbitrarily but only assuring that each machine gets no more than $2$ jobs;
  \item For each machine $p$, increase $b_p$ by the sum of job(s) scheduled to $p$ and update $\mathcal{F}$ accordingly;
  \item The rest parts of $h(\Lambda)$ are the same as $\Lambda$\enspace.
  \end{enumerate}
  By doing this, we obtain $h(\Lambda)$ in which all jobs have been scheduled and $h(\Lambda)$ is feasible\enspace.
  \qed
\end{theorem}

By Theorem \ref{main_theorem1}, as long as $M^* \geq 2k^*$, we can always schedule a \emph{free} job to a \emph{free} machine but without increasing its length. If $M^* < 2k^*$, Theorem \ref{main_theorem2} and \ref{main_theorem3} guarantee we still can make our decision in a fairly simple way. We present our algorithm using \emph{Iterative Rounding Method}, $IRA$, in Algorithm \ref{ira}.

\begin{algorithm}
\caption{$IRA$}\label{ira}
\begin{algorithmic}[1]
\REQUIRE An IP $\Delta$
\ENSURE A feasible integral solution $\mathcal{F}$
\STATE Construct natural linear programming relaxation $\Gamma$\enspace;
\STATE Solve $\Gamma$ optimally and let $y = (y_1, y_2, \dots, y_k)$ be the lengths of machines in the optimal solution;
\STATE Construct a BLPR $\Lambda$, letting $y$ be the \emph{upper bounding vector} and $\mathcal{F} = \emptyset$\enspace; \label{initialization}
\WHILE {$M^* > 0$}
  \IF {$M^* \geq 2k^*$}
    \STATE $\Lambda \gets f(\Lambda)$\enspace; \label{case1}
  \ELSE [$M^* < 2k^*$]
    \IF {there exists a machine $p$ with \emph{free} capacity of $1$}
      \STATE $\Lambda \gets g(\Lambda)$\enspace; \label{case2}
    \ELSE [every \emph{free} machine have more than $1$ capacity]
      \STATE $\Lambda \gets h(\Lambda)$\enspace; \label{case3}
    \ENDIF
  \ENDIF
\ENDWHILE
\RETURN $\mathcal{F}$\enspace;
\end{algorithmic}
\end{algorithm}

Finding a basic solution to a linear program can be done in polynomial time by using the ellipsoid algorithm \cite{Kha:79} then converting the solution found into a basic one \cite{Jai:01}. Together with the following observation
\begin{lemma}
  At Line \ref{initialization}, $y$ is a \emph{feasible upper bounding vector} of $\Lambda$\enspace.
  \qed
\end{lemma}
the correctness of $IRA$ follows from Theorem \ref{main_theorem1}, \ref{main_theorem2} and \ref{main_theorem3}.
\begin{corollary}
  Algorithm $IRA$ always terminates in polynomial time.
  \qed
\end{corollary}

The analysis of the performance of $IRA$ is simple with the help of \emph{upper bounding vector} $b$, noting that once a component of $b$ is increased, the machine will be no longer \emph{free}. We now show that $IRA$ is a $3$-approximation algorithm.
\begin{theorem}
  $IRA$ is a $3$-approximation algorithm.
\end{theorem}
\begin{proof}
  Consider any machine $p$ with length $\mathcal{A}$ in the solution produced by $IRA$\enspace. Note that once $b_p$ is increased at Line \ref{case2} machine $p$ will no longer be \emph{free}. So exactly one of the following statements is true when the algorithm terminates:
  \begin{enumerate}
  \item $b_p$ hasn't been increased, then $b_p \leq y_p$\enspace;
  \item $b_p$ has been increased once at Line \ref{case2}, then $b_p \leq y_p + t_q$ for some $q$\enspace;
  \item $b_p$ has been increased once at Line \ref{case3}, then $b_p \leq y_p + t_{q_1} + t_{q_2}$ for some $q_1,q_2$\enspace.
  \end{enumerate}

  Note that $y_p$ and $\max_q\{t_{q}\}$ are two trivial lower bounds of $\mathcal{OPT}$\enspace. Also note that after the algorithm terminates, the integral solution produced by $IRA$, contained in $\mathcal{F}$, is also bounded by \emph{upper bounding vector} $b$\enspace. By definition of \emph{feasible upper bounding vector}, we have inequality
  \begin{equation}
    \mathcal{A} \leq b_p \leq 3\mathcal{OPT}
  \end{equation}
  as expected.
  \qed
\end{proof}

\section{Conclusion}

In this paper, we consider the SMCC problem, a uniform variation of general scheduling problem, which has capacity constraints on identical machines. Using an extension of \emph{Iterative Rounding Method} introduced by Jain \cite{Jai:01}, we obtain a $3$-approximation algorithm. This is the first attempt to use \emph{Iterative Rounding Method} in scheduling problem and it shows the power of \emph{Iterative Rounding Method}. It is still unknown that whether the approximation ratio can be improved or whether the \emph{Iterative Rounding Method} can be used to obtain a good approximation algorithm for the non-uniform version of scheduling problem with capacity constraints.


\begin{thebibliography}{[MT1]}

\bibitem{Ang:Bam:Kon:01}
Angel, E., Bampis, E. and Kononov, A.:
A FPTAS for Approximating the Unrelated Parallel Machines Scheduling Problem with Costs.
In: 9th Annual European Symposium on Algorithms, pp.194-205. August 28-31, (2001).
LNCS, vol. 2161, Springer-Verlag, London, 194-205, (2001)

\bibitem{Gai:Mon:Woc:07}
Gairing, M., Monien, B., and Woclaw, A.:
A faster combinatorial approximation algorithm for scheduling unrelated parallel machines.
Theor. Comput. Sci. 380, 1-2 (Jun. 2007), 87-99. (2007)

\bibitem{Hoc:Shm:87}
Hochbaum, D. S., and Shmoys, D. B.:
Using dual approximation algorithms for scheduling problems: theoretical and practical results,
J. ACM 34, 144-162, (1987)

\bibitem{Jai:01}
K. Jain:
A factor 2 approximation algorithm for the generalized Steiner network problem,
Combinatorica, 21, pp.39-60, (2001)

\bibitem{Kha:79}
L. G. Khachiyan.:
A polynomial algorithm for linear programming (in Russian).
Doklady Akademiia Nauk USSR 244, 1093–1096, (1979).
A translation appears in: Soviet Mathematics Doklady 20, 191–194, (1979)

\bibitem{Len:Shm:Tar:90}
Lenstra, J. K., Shmoys, D. B., and Tardos, \'{E}.:
Approximation algorithms for scheduling unrelated parallel machines.
Math. Programming 46, 259-271, (1990)

\bibitem{Moh:Lap:07}
Mohit Singh, and Lap Chi Lau.:
Approximating minimum bounded degree spanning trees to within one of optimal.
In: The thirty-ninth annual ACM Symposium on Theory of Computing, San Diego, California, USA, pp.661-670, (2007)

\bibitem{Tsai:92}
L. H. Tsai.:
Asymptotic Analysis of an Algorithm for Balanced Parallel Processor Scheduling.
SIAM Journal on Computing, 21, pp.59-94, (1992)

\bibitem{Woe:05}
Gerhard J. Woeginger.:
A comment on scheduling two parallel machines with capacity constraints.
Discrete Optimization, Volume 2, Issue 3, September 2005, pp.269-272, (2005)

\bibitem{Yang:Ye:Zhang:03}
Heng Yang, Yinyu Ye, and Jiawei Zhang.:
An approximation algorithm for scheduling two parallel machines with capacity constraints.
Discrete Applied Mathematics, Volume 130, Issue 3, pp.449-467, (2003)

\bibitem{Zhang:Ye:01}
Jiawei Zhang, and Yinyu Ye.:
On the Budgeted MAX-CUT problem and its Application to the Capacitated Two-Parallel Machine Scheduling.
Working Paper, Department of Management Sciences, The University of Iowa, (2001)

\bibitem{Ets:Mat:85}
Maximilian M. Etschmaier, and Dennis F. X. Mathaisel.:
Airline Scheduling: An Overview.
Transportation Science, Vol.19, No.2, May 1985, pp.127-138, (1985)

\bibitem{Rus:Hof:Pad:95}
R.A. Rushmeier, K.L. Hoffman, and M. Padberg.:
Recent Advances in Exact Optimization of Airline Scheduling Problems.
Technical Report, George Mason University, (1995)

\end{thebibliography}
\end{document}